\documentclass[12pt,english]{article}
\usepackage{amssymb}
\usepackage{amsmath}
\usepackage{xcolor}
\usepackage{amsfonts}
\usepackage{graphicx}
\usepackage[T1]{fontenc}
\usepackage[latin9]{inputenc}
\usepackage{geometry}
\usepackage[normalem]{ulem}
\usepackage{amsmath}
\usepackage{amsthm}
\usepackage[authoryear]{natbib}
\usepackage{babel}
\usepackage{babel}
\usepackage{setspace}
\usepackage{url}
\usepackage{hyperref}

\newtheorem{claim}{Claim}

\newtheorem{lemma}{Lemma}

\newtheorem{proposition}{Proposition}

\hypersetup{
     colorlinks   = true,
     citecolor    = blue,
     linkcolor = blue
}

\begin{document}

\title{Competition for being visited first and ordered search deterrence%
\thanks{%
The authors thank Mark Whitmeyer and the audience of Theory Pre-Conference
to the 2025 ESWC, Seoul, Korea for questions, comments and suggestions.}}
\author{Wojciech Olszewski and Yutong Zhang\thanks{%
Department of Economics, Northwestern University; wo@northwestern.edu and
zhangyutong2017@u.northwestern.edu, respectively}}
\date{December 2025}
\maketitle

\begin{abstract}
When customers must visit a seller to learn the valuation of its product, sellers potentially benefit from charging a lower price on the first visit and a higher price when a buyer returns. \cite{armstrong2016search} show that such price discrimination can arise in equilibrium when buyers learn a seller's pricing policy only upon visiting. We depart from this assumption by supposing that sellers commit to observable pricing policies that guide consumer search and buyers can choose whom to visit first. We show that no seller engages in price discrimination in
equilibrium.
\end{abstract}

\section{Introduction}

Purchasing goods often requires search. People rarely buy homes without inspection or shoes without trying them on. Before arranging home improvements, they usually consult multiple contractors; when booking flights, consumers compare options across several websites. In many such settings, a seller can record whether a buyer is visiting for the first time or returning after an initial visit without purchase. In settings where recording is possible, it seems that a seller can benefit from charging a lower price on the first visit and a higher price upon return. Indeed, the greater price upon return discourages search, i.e., checking other options by visiting other sellers. 
Such pricing policies are often referred to as buy-now discounts.

We indeed observe buy-now discounts in practice. Anyone considering home improvements has likely been offered a special discount contingent on signing a contract during the initial consultation. And many people use multiple devices when shopping for flights in order not to lose a good price by leaving a website before making the purchase. There is also, mainly anecdotal, evidence from other markets, e.g., health clubs, car and house rentals, and labor markets. In contrast, buy-now discounts are typically not offered by real-estate agents or stores, even in cases in which stores could record visits.

Our paper highlights an important difference between these two settings. Inducements for quick decisions are offered by salesmen or contractors casually during the course of one-to-one encounters. Committing to a certain pricing policy (or building a reputation thereof) seems rather difficult in such settings. It seems equally difficult for websites selling air tickets, where prices must quickly respond to changing supply and demand. Commitments are much easier to sustain, and reputations are much easier to build for stores and real-estate agencies.

More specifically, we propose a model with two sellers (firms) and one buyer. Each seller offers a product for sale and the buyer is interested in purchasing only one of the two products. Firms post and commit to pricing policies\footnote{We argue later that commitment can also be viewed as a shorthand for reputation for a certain pricing policy.} (that is, a price for buying during the first visit, and a price for returning customers), and the buyer, after observing the posted prices, decides which firm to visit first. The buyer learns her valuation of the product offered by the visited seller. Then she chooses whether to purchase the product at the price for first comers, or to visit the other seller and learn her valuation of the other product. If the buyer visits the other seller, and so learns both valuations, she has 3 options: (i) return to purchase from the first seller at the return price, (ii) purchase from the second seller at the first-visit price, or (iii) leave the market. Our main result establishes that no firm engages in price discrimination in a symmetric equilibrium when search costs are zero. Moreover, this is the only symmetric equilibrium in pure strategies at least for uniform distribution of valuations. Numerical analysis suggests that such an equilibrium exists also for small positive costs.

Our finding contrasts with \cite{armstrong2016search}. In the primary version of their model, one seller posts two prices (one for buying during the first visit, and one for returning customers). A buyer can either buy from the seller or check her valuation of an exogenous outside option, and then decide whether to return or accept the outside option. In equilibrium, the seller engages in price discrimination against returning customers. \cite{armstrong2016search} note that the most natural outside option is buying from another seller. So they propose a duopoly version of their model, but find solving the model analytically intractable. They solve the model numerically for the uniform distribution of valuations, which confirms their finding that sellers offer in a symmetric equilibrium higher prices upon return to deter search.\footnote{We emphasize their numerical result for duopolies and uniform distribution of valuations, because it sharply contrasts with our findings. However, \cite{armstrong2016search} also analytically derive the equilibrium conditions under a general distribution of valuations. In addition, their paper contains other results, e.g., they study more general selling mechanisms beyond simple pricing policies.} As noted in the working paper version \cite{armstrong2013search}, even when there is no cost of visiting firms, they estimate the equilibrium buy-now discount for 12\% of the price for returning customers.

The difference between their and our findings comes from the assumption that their buyer randomly chooses the seller for her first visit, while our buyer first learns the pricing policies the sellers commit to, and then chooses which seller to visit first.
\footnote{More precisely, \cite{armstrong2016search} assume that the buyer learns a seller's pricing policy only when she visits the sellers. During this visit, the
seller also commits to the price upon return. We emphasize, however, that the key different component of our model is the sellers' commitment to (not only the observability of) the pricing policies
before the buyer decides which seller to visit first.
\par
In the \cite{armstrong2016search} setting, the buyer correctly anticipates the sellers' prices in equilibrium, which makes no difference compared to physically observing them. However, no seller can effectively deviate in their setting even by posting different prices when the seller cannot credibly commit to these different posted prices, because the buyer knows that once she visits the seller it will be in the seller's interest to charge the equilibrium prices. Therefore, the seller who cannot commit to pricing policy in advance cannot affect the buyer's decision which seller to visit first.}
Since \cite{armstrong2016search} show that buy-now discounts reduce consumer surplus and total welfare, our results demonstrate that commitment is desirable from the welfare perspective. The same holds for building reputation. Therefore any policy that facilitates commitment or building reputation, e.g., by promoting transparency of pricing strategies, will enhance welfare.

Our results are easy to understand when there are no exogenous search costs. If sellers offer buy-now discounts, then each of them has an incentive to commit to a discount also upon return to attract buyers to visit its location first. Indeed, the buyer prefers to visit the seller who offers a discount also upon return in order not to lose the discount from the seller who offers it only during the first visit. This competition for being visited first makes the buy-now discount ineffective, and guarantees the existence of symmetric equilibria in which firms post uniform prices. However, this argument does not give the uniqueness of such equilibria, because for some pairs of prices, each seller prefers to be visited second. We have been able to exclude symmetric equilibria in which sellers conduct price discrimination across visits only under zero search cost and the uniform distribution of valuations, 
because in these cases each seller prefers to be visited first for all pairs of prices. 

If search costs are positive, under the uniform distribution of valuations, sellers have an additional incentive for reducing the price gap in order to encourage customers to visit themselves first. And for small costs, by upper hemi-continuity of equilibria, the discrimination must be small, if any. Our numerical analysis suggests that (for small costs) there exist equilibria in which firms do not discriminate.

The rest of the paper is organized as follows. Section \ref{s11} contains a short literature review. In Section \ref{s2}, we introduce the model, with Section \ref{s21} devoted to a discussion of the commitment assumption and reputation. Section \ref{s3} contains a heuristic explanation of contrasting results in the models with and without commitment. In Section \ref{s4}, we present our results when the search cost is zero. In Section \ref{s5}, we discuss the case with positive search cost; in particular, Section \ref{s53} contains a summary of our numerical findings. The derivations and code used in the numerical analysis are provided in an appendix (to be posted online). We conclude in Section \ref{s6}.

\subsection{Literature review}\label{s11}

Our paper contributes to the literature on search deterrence. The most closely related work is \cite{armstrong2016search}. In that paper, they study a random search model and discuss sales techniques that discourage consumer search, such as making it more expensive for returning visits. They find that in the symmetric duopoly setting where the buyer's valuations of the firms' products are uniformly distributed, firms can benefit from price discrimination of returning customers. In our model, we still allow for the same deterrence strategy, but we assume the buyer can choose which firm to visit first after observing costlessly the prices to which firms commit. This gives us qualitatively different results from \cite{armstrong2016search}. \cite{armstrong2016search} have objectives and results that go beyond our more narrow focus on buy-now discounts. For example, they also study more general selling mechanisms beyond simple pricing policies. \cite{zhu2012finding} studies, in the context of financial markets, a closely related motive for price discrimination of returning customers, namely, a party that returns signals that it has checked other options and it has not found them attractive. This induces an offer that is less attractive than the initial one. In the uniform valuation setting of \cite{armstrong2016search}, \cite{groh2022search} study yet another reason for price discrimination. Each firm can reveal to its opponent that it has been visited by a buyer. The opponent can offer prices based on buyer's search history. They show that firms only have incentives to disclose search history when search costs are low. In \cite{zhu2012finding} and \cite{groh2022search}, as opposed to our paper,  sellers cannot commit to their pricing strategies.


A handful of papers (see, e.g., \cite{ellison2012search}, \cite{gamp2016guided} and \cite{petrikaite2018consumer}) study search deterrence in the context of obfuscation. In such models, aside from prices, firms also set search costs. In \cite{ellison2012search}, firms have incentives to raise the time of learning their prices when the search cost is a convex function of the total time of searching. In \cite{gamp2016guided} and \cite{petrikaite2018consumer}, a monopolist selling multiple products can increase its profit by influencing consumers' order of search through varying the costs of learning the values of its products.

Our work is embedded in the literature on price-directed search. Such models, building on \cite{weitzman1979optimal}, are studied in \cite{armstrong2017ordered} and \cite{choi2018consumer}. Other related papers study consumer search guided by the informativeness of signals provided by firms. \cite{au2023attraction}, abstracting away from price competition, show that firms provide the first-best level of information only if the quality of their products is sufficiently high or the competition is sufficiently fierce. \cite{au2024attraction} show that when both posted prices and information policies can guide consumer search, the equilibrium exhibits price dispersion and active search. Our paper instead focuses on how commitment to pricing policies that guide consumer search renders search deterrence unprofitable in a duopoly setting. Some papers study pricing games between sellers with exogenous search orders determined by location or prominence, for example, \cite{arbatskaya2007ordered}, \cite{armstrong2009prominence} and \cite{zhou2011ordered}, while in our paper, the search order is endogenously determined by firms' pricing strategies.

Finally, our paper is also related, albeit more distantly, to the large literature on intertemporal price discrimination outside search models, especially the work on advance-purchase discounts when buyer uncertainty is resolved over time (e.g., \cite{nocke2011advance}), and the discussion of the optimality of such discounts (e.g. \cite{gale1992efficiency}, \cite{dana2001monopoly}) versus uniform pricing (e.g., \cite{riley1983optimal}; \cite{wilson1988optimal}).

\section{Model}\label{s2}

Two symmetric firms (sellers) $i=1,2$ each have one unit of a product for sale to a buyer (she), who is interested in purchasing only one unit from either seller. The buyer's valuation of the unit that seller $i=1,2$ has is $v_{i}$. She initially does not know any of her values $v_{i}$ and only knows that $v_{i}$ is distributed according to $F$, which is supported on \([0,1]\). The buyer can learn the exact value $v_{i}$ by visiting firm $i$. For most of the paper, we assume that the buyer incurs no cost of visiting a firm; in Section \ref{s5}, we explore the case when this cost is positive.

The timing of the game is as follows: Each seller simultaneously posts prices $\underline{p}_{i}\leq \overline{p}_{i}$ for its good. The lower price stands for the amount the buyer has to pay if she decides to acquire the product from seller $i$ during her first visit at that seller's location; the higher price represents the amount that she has to pay if she leaves the seller's location and acquires the product upon return.\footnote{We assume upfront that $\underline{p}_{i}\leq  \overline{p}_{i}$\, because if $\underline{p}_{i}>\overline{p}_{i}$, the buyer could always reject buying during the first visit, and return (after visiting the other seller or not) to buy at a lower price.} That is, the seller can record the buyer's visits at its location, or at least whether she visits for the first time or returns. In contrast, the seller cannot observe whether the buyer has visited the other seller. Sellers are fully committed to the prices they post.

Next, the buyer decides which seller to visit first. During this visit, the buyer learns $v_{i}$, and she can choose to buy the product from seller $i$ by paying $\underline{p}_{i}$. If she does, this ends the game. The buyer can also choose to leave seller $i$'s location and visit the other seller $j$ to learn $v_{j}$. The buyer can now buy the product from seller $j$ at $\underline{p}_{j}$, return to seller $i$ and buy its product at $\overline{p}_{i}$, or leave the market without purchasing either product. Each of these decisions ends the game. At any moment, the buyer can stop searching. Returning to seller $i$ incurs no additional cost. This last assumption is consistent with prior literature. Similar to \cite{armstrong2016search}, we acknowledge that typically buyers face some (usually much smaller) cost of returning, but this cost is unlikely to change any our conclusions, and would only blur the comparison with \cite{armstrong2016search}.

We normalize the production cost to zero and a seller's payoff is the price that the buyer pays if she decides to buy from this firm.  The buyer's payoff is equal to the value of the acquired product minus the price paid for it and the search cost incurred, which is zero except in Section \ref{s5}, or zero if no product is acquired.

\subsection{Discussion of the commitment assumption}\label{s21}

The difference between \cite{armstrong2016search} and our findings comes from the assumption that their buyer chooses randomly which seller to visit first (or, as they put, the buyer learns the pricing policy only after visiting a seller), while our buyer first learns the pricing policies sellers commit to, and then chooses which seller to visit first. So, commitment is the key assumption for our analysis.

Commitment to uniform pricing is common in stores, but commitment may also be a feature of pricing policies in the long run, even when there is no explicit available instrument that would allow firms to commit. Indeed, if the model is interpreted literally, the interaction between the buyer and the seller lasts for only one period. In this case, sellers have an incentive to offer buy-now discounts to deter search, and \cite{armstrong2016search} show that they indeed do so in equilibrium. The buyer correctly anticipates that she will be offered a discount, but it has no impact on the decision over which seller to visit first, because both sellers will do the same.

However, if the model is interpreted as a shorthand (or a reduced version) of a repeated interaction of the same sellers possibly with different buyers (a more realistic scenario in numerous settings), then buyers learn from the past pricing of the sellers. Intuitively, if a seller has been offering a discount for the first visit, the buyer who comes next expects that such a discount will be offered during their visit. If a seller typically charges a uniform price, the buyer expects no discount to be offered. The information regarding sellers' pricing policies usually quickly spreads out by word of mouth. And the buyer chooses which seller to visit first based on this expectation. Of course, the seller can surprise one or two buyers, but if it happens on a more regular basis, buyers will begin expecting the different policy.

More formally, a large literature on reputation shows that a long-lived seller can achieve a commitment outcome when interacting with a sequence of short-lived buyers. In the duopoly setting, we will show that at least for zero search cost and the uniform distribution of valuations, each seller benefits from committing to a uniform pricing policy in response to any discriminatory pricing strategy of the opponent. So, the reputation literature suggests that the uniform pricing strategy should be a unique equilibrium of the game between two long-lived sellers and a sequence of short-lived buyers.

Of course, each seller also has an incentive to use any available instrument that allows itself to commit to  uniform pricing in response to any pricing strategy of the opponent. So, if we explore the game in which the sellers first choose whether to commit to their pricing strategy, this game would have a unique equilibrium in which the sellers choose to commit to a uniform pricing (at least for zero cost and uniform distribution of valuations). 

\section{Example (bonus games)}\label{s3}

In this section, we attempt to explain in the simplest possible manner why discriminatory pricing is an equilibrium phenomenon when firms cannot commit to their pricing strategies, but uniform pricing is an equilibrium when firms can commit to their pricing strategies.  
To do so, we suggest a somewhat heuristic analysis that intuitively captures the difference between the two scenarios.

\subsection{Posting single prices}\label{s31}

Suppose first that the selling policies are restricted to posting a single price. This case is equivalent to the model of \cite{perloff1985equilibrium} with two firms. Then, the game with the uniform distribution of valuations has a unique symmetric equilibrium in which each firm posts $p^*=\sqrt[2]{2}-1$.

Take any price \(0<p<1\). Indeed, since there is no cost of visiting a firm, the buyer visits both firms and buys from the one in which her valuation is greater. Therefore each firm sells with probability $(1-p^{2})/2$. For $p$ to be an equilibrium price, it must be that no firm benefits from lowering or raising its price by $dp$. Suppose a firm lowers its price by $dp$ (or raises its price if $dp$ is negative).\footnote{For simplicity, we analyze in this example only the first-order effects.} Then the firm loses $(1-p^{2})dp/2$ by selling at a lower price to the buyer who would buy from the firm at $p$. However, it gains by: (1) selling to the buyer who would not buy at price $p$, and (2) ``stealing'' the buyer from the other firm.

The gain from (1) is $pdp$. The increase in demand from (2) is equal to the area of valuation pairs such that the valuation of the firm's product is greater than $p$, and the valuation of the other firm's product is the greater of the two but not by more than $dp$. The area of such valuation pairs in the unit square is equal to $(1-p)dp$. So, the gain from (2) is \(p(1-p)dp\). So in an equilibrium, it must be that 
\begin{equation*}
\frac{(1-p^{2})dp}{2}=p^{2}dp+p(1-p)dp.
\end{equation*}
This yields a quadratic equation for $p$; by solving it, we obtain
\begin{equation*}
p^*=\sqrt[2]{2}-1.
\end{equation*}

\subsection{First-visit bonuses}

Consider now the situation in which firms are allowed to conduct price discrimination, but only by offering small bonuses to first comers. More specifically, each firm can either offer the same price $p^*$ for all customers, or it can offer price $p^*-\Delta p$ for first comers and $p^*$ for returning customers. Do firms have an incentive to deter search during the first visit by offering the lower price? The answer is positive when firms cannot commit to the pricing policies, as stated in the following proposition:

\begin{proposition}[Firms cannot commit to pricing policies ex ante.]\label{p1}
Suppose that one of the two firms, say firm 1, offers price $p^*$ for all customers. Then, for any sufficiently small bonus $\Delta p$, the payoff of firm 2 is greater when it conducts price discrimination than when it offers $p^*$ for all customers.
\end{proposition}

Indeed, if firm 2 is visited first, then it makes a first-order gain (in $\Delta p$) when it engages in price discrimination compared to offering a uniform
price (Claim \ref{c1}). If firm 2 is visited second, then it generates a lower-order loss (Claim \ref{c2}).

\begin{claim}\label{c1}
If firm 2 is visited first, then the discriminatory pricing  results in a gain of order $\Delta p$ to firm 2 compared to uniform pricing.
\end{claim}

\noindent\textbf{Proof}: By offering the bonus, firm 2 sells to a buyer who could have obtained a higher surplus from firm 1's product but decides not to visit firm 1 in order not to lose the bonus. That is, the bonus deters some search. The pairs of valuations for which this happens are depicted in a green contour in Figure \ref{fig1}, where $v$ is the valuation at which the buyer is indifferent between buying from firm 2 on the first visit (and getting the bonus) and visiting firm 1 before making the decision from which firm to purchase.

\begin{figure}
    \centering
    \includegraphics[width=0.5\linewidth]{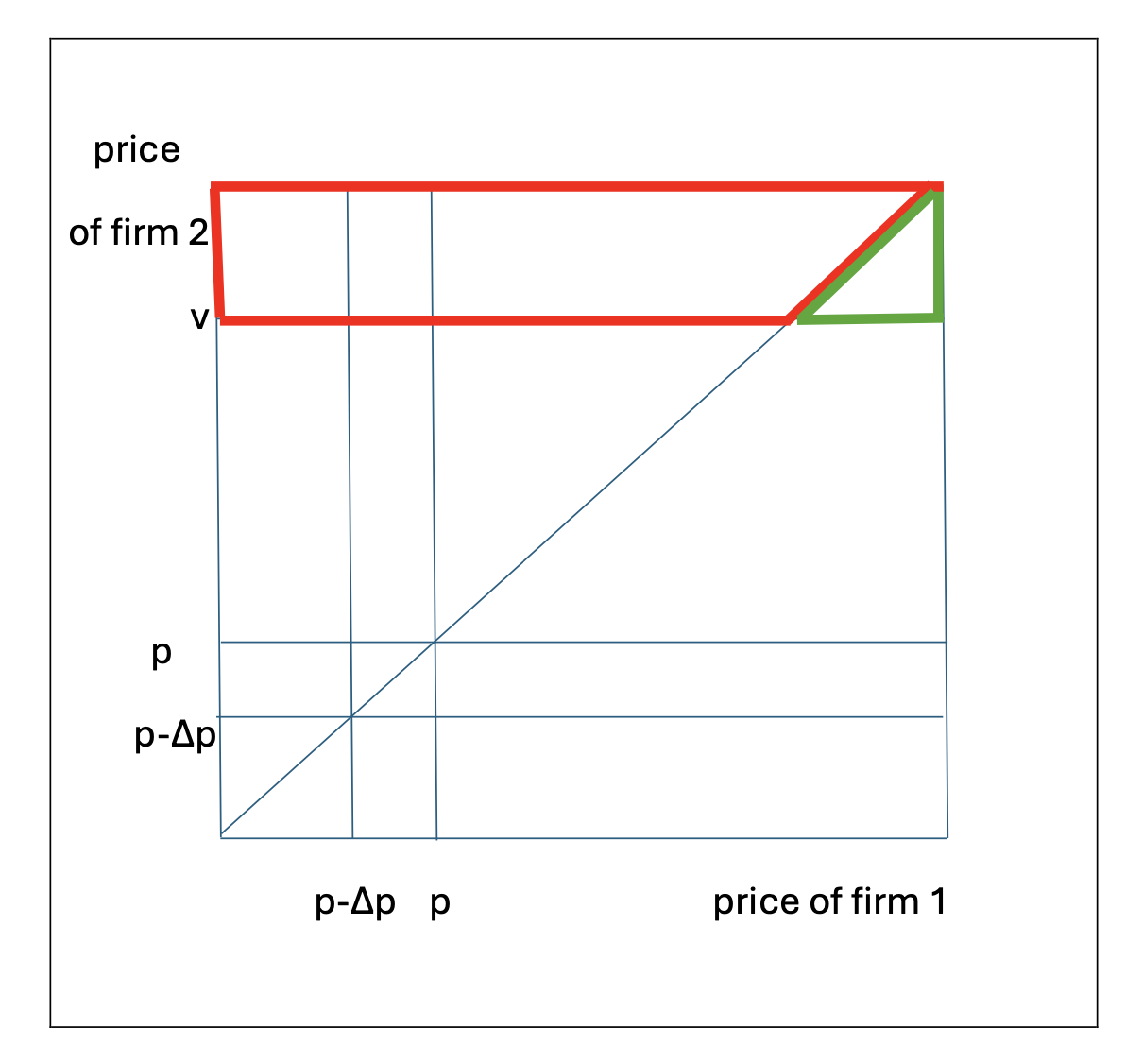}
    \caption{The buyer who visits firm 2 first and is offered a bonus will buy immediately if her valuation is above $v$ and otherwise visit firm 1 and buy the product she values more (provided that at least one valuation is higher than price $p$). If both firms always sell at price $p$, then the buyer visits both firms and buys from the one whose product she values more (provided that at least one valuation is higher than price $p$)}
    \label{fig1}
\end{figure}

So, the gain from deterring search is of the order of the probability of the green area. When $F$ has a bounded density  which is also bounded away from 0 on $[0,1]$, this probability is of the same order as what the buyer who values firm 2's product at $v$ expects to gain from visiting firm 1, which in turn must be equal to the loss of the bonus. Thus, firm 2's gain from deterring search is of order $\Delta p$. For uniform distribution of valuations, the buyer who values firm 2's product at $v$ expects to gain by visiting firm 1
\begin{equation*}
\int_{v}^{1}(x-v)dx=\frac{1}{2}(1-v)^{2},
\end{equation*}
and this is also the green area and probability thereof.

Firm 2 loses from offering the bonus by charging $\Delta p$ less the buyer who would buy from firm 2 under uniform pricing. The pairs of valuations for which this happens are depicted in Figure \ref{fig1} in a red contour. However, this loss is of order lower than $\Delta p$, because $v$ tends to 1 when $\Delta p$ tends to 0, so the probability of the red area vanishes with $\Delta p$, and the expected loss is equal to the probability of the red area multiplied by $\Delta p$. This completes the proof of the claim. \qed

\begin{claim}\label{c2}
If firm 2 is visited second, then the discriminatory pricing results in a loss of order lower than $\Delta p$ compared to uniform pricing.
\end{claim}

\noindent\textbf{Proof:} As an auxiliary step, return to the setting in which firms are not allowed to offer different prices. Then, charging price $p-\Delta p$ against price $p$ results in a loss, because $p$ is a symmetric equilibrium price. However, this loss is of a lower order than $\Delta p$. In this auxiliary setting, the loss is the sum of two components: the effect on firm 2's payoff when it is visited first and the effect on its payoff when it is visited second. Note that these two effects are exactly equal, because in any case the buyer visits both firms and the decision over which firm's product to purchase does not depend on which firm is visited first.

Consider now that firms are allowed to conduct price discrimination. If firm 2 is visited second, then only its price for first comers matters, and the effect of firm 2's price for first comers on its payoff is exactly the same as that in the auxiliary scenario when firm 2 is visited as second. This last effect is a loss of order lower than $\Delta p$, which completes the proof of the claim. \qed

If firms cannot commit to their pricing policies, the probability that the buyer visits each firm first is \(\frac{1}{2}\). This, together with Claims \ref{c1} and \ref{c2}, completes the proof of Proposition \ref{p1}.

Suppose now that firms first commit to their pricing policies. Then, the answer to the question asked at the beginning of this section is negative.

\begin{proposition}[Firms commit to pricing policies ex ante.]\label{p2}
If no bonus is offered by firm 1, then offering a bonus by firm 2 reduces the profit of firm 2 compared to uniform pricing, 
\end{proposition}

Indeed, if the bonus is offered, the buyer first visits the other firm, therefore it does not deter search. And its effect is the same as in the auxiliary setting in which firms are not allowed to discriminate, conditional on firm 2 being visited as second. However, we have already established that this results in a loss (of an order lower than $\Delta p$) compared to uniform pricing. This yields Proposition \ref{p2}.

\bigskip
\noindent\textbf{Remark 1.} One can show that Propositions \ref{p1} and \ref{p2} remain true even when firm 1 conducts price discrimination. So, the unique symmetric equilibrium of the game in which each firm can only charge $p$ or $p-\Delta p $ has: (i) both firms engage in price discrimination if firms cannot commit, and (ii) no firm conduct price discrimination when firms commit to their pricing policies.

\bigskip
Of course, the analysis of this example is heuristic. It assumes an exogenous price $p$ and an exogenous small bonus $\Delta p$.  However, \cite{armstrong2016search} show that when firms can choose any prices they want and the buyer's decision about which firm to visit first is exogenous (because the buyer learns a firm's pricing policy only when she visits the firm), there is a symmetric equilibrium in which firms offer a lower price on the first visit and a higher price upon return. And we show that when firms can commit upfront to any price pairs they want, in the \textit{unique symmetric equilibrium}, each firm charges the same price $p=\sqrt[2]{2}-1$ no matter whether this is the first visit or return.

\section{Results for \(c=0\)}\label{s4}

This section contains our main results when the cost of visiting any seller is negligible. Let $p^{\ast }$ be the price in a symmetric equilibrium of the game in which firms are allowed to post only a single price for both first comers and returning customers.

Consider now the game in which firms are allowed to post any, possibly different prices, that is, the strategies are pairs of prices such that $\underline{p}_{i}\leq \overline{p}_{i}$; of course, prices can be different for different (though symmetric) firms.

The prices \underline{$p$}$_{i}=\overline{p}_{i}=$\underline{$p$}$_{-i}=\overline{p}_{-i}=p^{\ast }$ for $i=1,2$ are still equilibrium prices. Indeed, no firm $i$ can profitably deviate to \underline{$p$}$_{i}^{\prime }<\overline{p}_{i}^{\prime }$. Observing such a deviation, the buyer would first visit the other firm, and only the lower price \underline{$p$}$_{i}^{\prime }$ of the deviating firm would ever matter. So, only deviations to a single price, i.e., to \underline{$p$}$_{i}^{\prime }=\overline{p}_{i}^{\prime }\neq p^{\ast }$ have a chance to be profitable. However, because \underline{$p$}$_{i}=\overline{p}_{i}=p^{\ast }$ for $i=1,2$ is an equilibrium of the game in which each firm is allowed to post only a single price, such deviations are not profitable either. This yields the following result:

\begin{proposition}\label{p3}
Each firm posting a single price \underline{$p$}$_{i}=\overline{p}_{i}=p^{\ast }$  is an equilibrium of the game in which firms post pairs of prices $\underline{p}_{i}\leq \overline{p}_{i}$.
\end{proposition}

When the distribution of valuations is uniform, we know from Section \ref{s31} that the game in which firms are allowed to post only a single price has a unique symmetric equilibrium. Under the same distribution, this  equilibrium  remains the unique symmetric equilibrium of the game that allows for price discrimination. Thus, at least in this case, an equilibrium in which firms discriminate does not coexist with the equilibrium in which they charge the uniform price. This is in sharp contrast with \cite{armstrong2016search}, in whose model with no commitment firms discriminate in equilibrium even when \(c=0\). We formally state this result as follows:

\begin{proposition}\label{p4}
Under uniform distribution of valuations, \underline{$p$}$_{1}=\overline{p}_{1}=$\underline{$p$}$_{2}=\overline{p}_{2}=p^{\ast }=\sqrt[2]{2}-1$ is a unique symmetric equilibrium.
\end{proposition}

Intuitively, a pair $\underline{p}<\overline{p}$, where $\overline{p}\geq 1$ prohibits any returning, cannot be a symmetric equilibrium, since a unilateral deviation to the uniform price of $\underline{p}$ makes the total profit greater, and the total profit is shared equally between the two sellers, as is the profit in the putative equilibrium $\underline{p}<\overline{p}$. So, suppose that $\overline{p}<1$. Then, $\underline{p}<\overline{p}$ can be a symmetric equilibrium strategy only if each seller's profit would be the same regardless of whether she is visited  first or  second. Otherwise, any seller could profitably deviate by increasing or decreasing $\overline{p}$, depending on whether the profit when visited first falls below or exceeds that obtained when visited second. Thus, one can compare the profit under unilateral deviation to a uniform price of $\underline{p}$, in which case the deviating firm is always visited first, with that under the putative equilibrium contingent on this firm being visited second. 

This deviation has some positive and some negative effect on the profit of the deviating seller, say seller 1: (a) the seller now sells to the buyer who values its product more, but who would be deterred from searching by the pricing policy of seller 2, and who would not visit seller 1 as second; (b) the seller loses the buyer who values its product less but who does not return to seller 2, after visiting its location as first, because of the greater price of seller 2 upon return. For some distributions of valuations and some pairs of prices $\underline{p}<\overline{p}$, the gain from (a) is smaller than the loss from (b), but the former effect always dominates the latter for the uniform distribution.


\bigskip

\noindent\textbf{Proof:} Suppose there is an equilibrium in which $\underline{p}_{1}=\underline{p}_{2}<\overline{p}_{1}=\overline{p}_{2}$. We will argue that each firm can profitably deviate to posting the smaller price for all customers.

Observe first that when firms  offer different prices, the buyer's search strategy can be characterized by cutoffs $v_{i}^{\ast }$. If she visits firm $i$ first, she will stop search and buy from this firm if $v_{i}\geq v_{i}^{\ast }$, and continue search and pick the product with a higher surplus otherwise. Under symmetric strategies, $v_{i}^{\ast}=v_{-i}^{\ast }:=v^{\ast }$ and the buyer is indifferent between which firm to visit first. For all $\underline{p}\leq \overline{p}\in (0,1)$, $v^{\ast }$ is the solution of the following equation:
\begin{equation}
\int_{0}^{1}\max \{x-\underline{p},v^{\ast }-\overline{p},0\}dF(x)=v^{\ast }-%
\underline{p},  \label{*}
\end{equation}
Note that $v^{\ast }>\underline{p}$, but it can happen that $v^{\ast }\leq \overline{p}$, in which case the buyer (almost) never returns to the firm visited first.

\bigskip

\noindent\textbf{Step 1:} We first state a lemma that connects a firm's profit from being visited first and second:

\begin{lemma}\label{order} 
Suppose that: (a) each firm offers two different prices, (b) the buyer is indifferent between which firm to visit first, and (c) firms' prices at the first visit are lower than $1$ (the upper bound on the buyer's value). Then, each firm can guarantee itself to be visited first by reducing its price upon return.

If, in addition, $\overline{p}_{i}<v_{i}^{\ast }$ (that is, there is a strictly positive probability that the buyer after visiting firm $i$ first will return to the firm), then firm $i$ can also guarantee itself to be visited second by
raising its price upon return.
\end{lemma}

\begin{proof}
Note first that by changing its price upon return, the firm does not affect the buyer's payoff if the buyer visits the firm's opponent first. So, suppose that the buyer visits first the firm that reduces its price upon return. Let it be firm 2. With no loss of generality, assume that $\overline{p}_{2}\leq 1$.\footnote{This is without loss of generality, because $\overline{p}_{2}>1$ and $\overline{p}_{2}=1$ lead to exactly the same outcome.} We claim that then the buyer's payoff is greater when the price is $\overline{p}_{2}^{\prime }<$ $\overline{p}_{2}$ than when it is $\overline{p}_{2}$. Indeed, suppose that the price upon return is $\overline{p}_{2}$, and consider an optimal continuation strategy of the buyer after she decides to visit firm 2 as first. If the buyer uses the same continuation strategy when the price drops to $\overline{p}_{2}^{\prime }$, then the payoff of the buyer is no smaller because she will pay less if she returns. However, an optimal continuation strategy when the price upon return is $\overline{p}_{2}^{\prime }$ gives a  payoff which is no smaller (it is actually strictly greater, because the event in which buyer returns in nontrivial). Therefore, visiting firm 2 first gives the buyer a higher expected payoff than visiting firm 1 first. An analogous argument implies that when $\overline{p}_{2}<1$, by raising this price to $\overline{p}_{2}^{\prime }>\overline{p}_{2}$ firm 2 reduces the buyer's payoff when she visits firm 2 first.\bigskip 
\end{proof}

\noindent\textbf{Step 2:} We first consider the scenario in which $v^{\ast}\leq \overline{p}$. Then each firm has an incentive to deviate to $\underline{p}_{i}^{\prime }=\overline{p}_{i}^{\prime }=\underline{p}$. To see why, note that before the deviation, the profit is less than $\frac{1-F^{2}(\underline{p})}{2}\underline{p}$ because there is a chance that a buyer with $\underline{p}<v_{i}\leq v^{\ast }$ will visit the other seller but learn that $v_{-i}<\underline{p}$. After the deviation, the seller will be visited first and the profit is exactly $\frac{1-F^{2}(\underline{p})}{2}\underline{p}$, which is strictly greater. Thus we can conclude that in equilibrium, $v^{\ast }\leq \overline{p}$ cannot happen.

\bigskip

\noindent \textbf{Step 3:} We are left with the case $v^{\ast }>\overline{p}$, in which there is a strictly positive probability that the buyer will return to the firm she visits first. We can invoke Lemma \ref{order} and claim that each firm's profit from the first or the returning visit should be the same. Now still consider the deviation of setting $\underline{p}_{i}^{\prime}=\overline{p}_{i}^{\prime}=\underline{p}$. We compare the profit after the deviation with the profit from being visited second before the deviation, and we show that the deviation is profitable if the following condition is satisfied:
\begin{equation}\label{**}
\int_{\underline{p}}^{\overline{p}}[F(x)-F(\underline{p})]dF(x)+\int_{%
\overline{p}}^{v^{\ast }}[F(x)-F(x-\overline{p}+\underline{p})]dF(x)
<\int_{v^{\ast }}^{V}[1-F(x)]dF(x).
\end{equation}

This is a condition on $F$ that is required for our argument. We show later that this condition is satisfied when $F$ is uniform. 

\begin{figure}[h]
    \centering
    \includegraphics[width=0.5\linewidth]{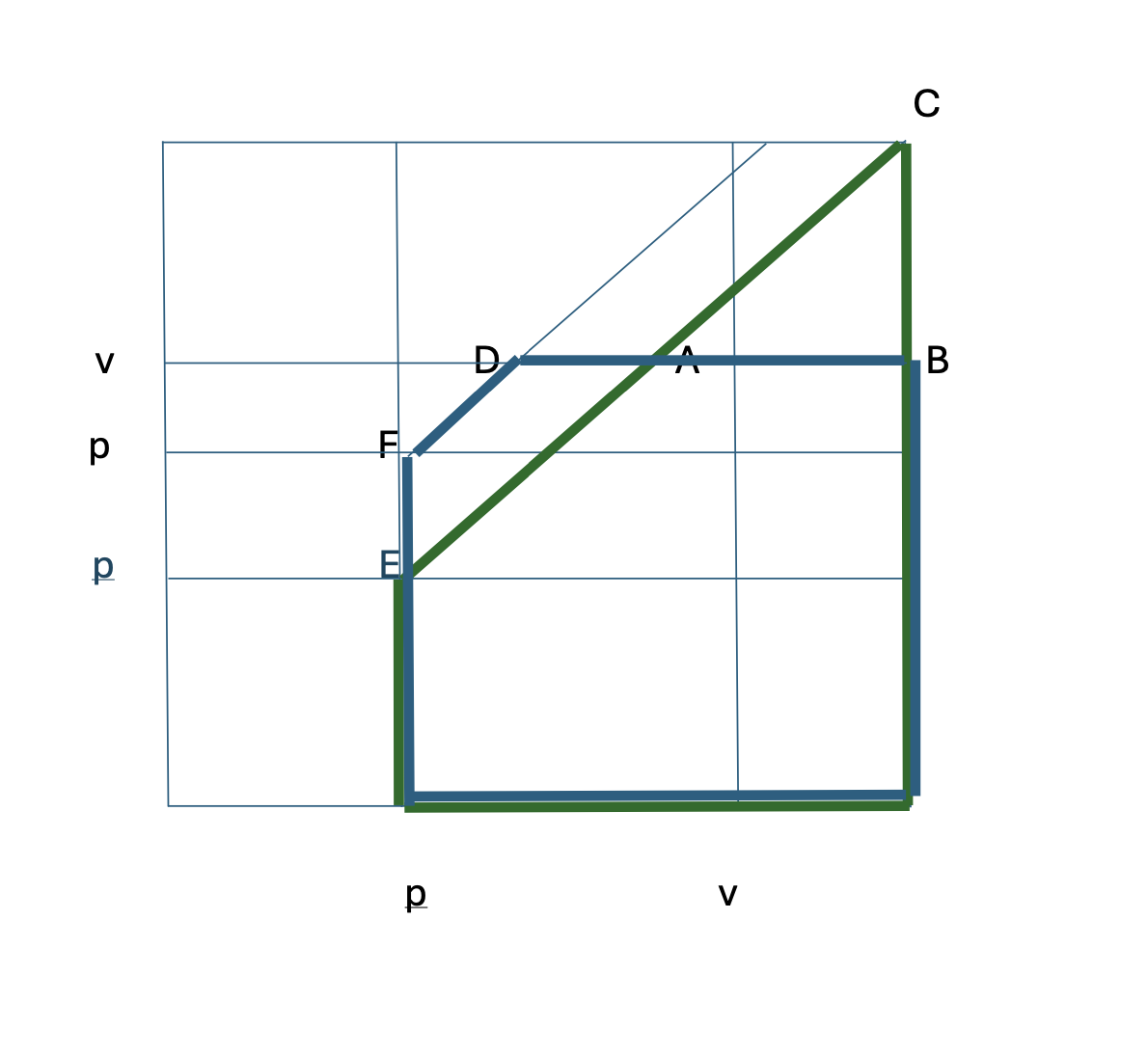}
    \caption{The demand from being visited second before the deviation is surrounded by the blue lines, and the demand from being visited first after the deviation is surrounded by the green lines.}
    \label{fig2}
\end{figure}

The reason for a unilateral deviation to  $\underline{p}_{i}^{\prime }=\overline{p}_{i}^{\prime }=\underline{p}$ being profitable is easier to see with Figure \ref{fig2}. In the figure, the demand from being visited second before the deviation is surrounded by the blue lines, and the demand from being visited first after the deviation is surrounded by the green lines. The deviation is profitable if the area of the trapezoid ADFE, which is the LHS of (\ref{**}), is strictly smaller than the area of the triangle ABC, which is the RHS of (\ref{**}). This completes the proof.

Finally, we will show that when \(F\) is the uniform distribution, Condition (\ref{**}) is satisfied. Indeed, the definition of $v^{\ast }$ (condition (\ref{*})) is equivalent to: 
\begin{equation*}
\int_{0}^{v^{\ast }-\overline{p}+\underline{p}}(v^{\ast }-\overline{p})dx+\int_{v^{\ast }-\overline{p}+\underline{p}}^{1}(x-\underline{p})dx=v^{\ast }-\underline{p},
\end{equation*}
That is, 
\begin{equation*}
(v^{\ast }-\overline{p})(v^{\ast }-\overline{p}+\underline{p})+\frac{1}{2}(1-\underline{p})^{2}-\frac{1}{2}(v^{\ast }-\overline{p})^{2}=v^{\ast }-\underline{p},
\end{equation*}
and after rearrangement, we get: 
\begin{equation}
\frac{1}{2}(v^{\ast })^{2}+v^{\ast }(\underline{p}-\overline{p}-1)+\frac{1}{2}\overline{p}^{2}-\overline{p}\underline{p}+\frac{1}{2}+\frac{1}{2}\underline{p}^{2}=0.  \label{***}
\end{equation}

And we can directly calculate the difference between the LHS and RHS of (\ref{**}) as follows: 
\begin{align*}
& \int_{\underline{p}}^{\overline{p}}[F(x)-F(\underline{p})]dF(x)+\int_{\overline{p}}^{v^{\ast }}[F(x)-F(x-\overline{p}+\underline{p})]dF(x)-\int_{v^{\ast }}^{1}[1-F(x)]dF(x) \\
=& \frac{1}{2}(\overline{p}-\underline{p})^{2}+(\overline{p}-\underline{p})(v^{\ast }-\overline{p})-\frac{1}{2}(1-v^{\ast })^{2}=-\frac{1}{2}(v^{\ast
})^{2}-v^{\ast }(\underline{p}-\overline{p}-1)+\frac{1}{2}\underline{p}^{2}-\frac{1}{2}\overline{p}^{2}-\frac{1}{2},
\end{align*}
and by (\ref{***}), we get $\underline{p}(\underline{p}-\overline{p})<0$,
which says the deviation is indeed profitable. Thus, we have shown that if $F$ is uniform, there exists no symmetric equilibrium in which each firm offers different prices. \qed

\bigskip
\noindent\textbf{Remark 2.} When condition (\ref{**}) is satisfied,  the unilateral deviation to the lower price when each firm offers two different prices is indeed profitable. However, even when this condition is violated, there may still be no symmetric equilibrium in which each firm offers two different prices, because there are other possibly profitable unilateral deviations. Unfortunately, checking the potential profitability of those other deviations turns out intractable.

\section{Results for \(c>0\)}\label{s5}

\subsection{Modeling cost of visiting sellers}

In this paper, we argue that the competition among firms for being visited first by customers may have striking effects on market outcomes. To demonstrate it, we contrast our result with that in \cite{armstrong2016search}. They numerically find search deterrence as a symmetric equilibrium phenomenon for the uniform distribution of valuations and an interval of search costs $c\geq 0$. We analytically find no search deterrence as the unique symmetric equilibrium for search cost $c=0$. 

Although the results for uniform distribution of valuations and $c=0$ provide a sufficient contrast, the reader (as well as us) would like to know whether our result can be extended to $c>0$. Intuitively, it seems that exogenous search costs should escalate the competition for being visited first, as customers want to minimize the number of times they pay these costs. However, like \cite{armstrong2016search}, we find the model analytically intractable, and we failed to solve it formally for $c>0$. Even when sellers are allowed to charge only single prices, our game in which firms first commit to their prices has no symmetric equilibrium in pure strategies. Indeed, the seller that sets a lower price is visited first, and when prices are positive, each seller profitably deviates by slightly reducing its price.\footnote{The former statement follows, because visiting first the seller who offers a lower price gives a better chance for saving the cost of visiting the seller's opponent. The latter statement follows because there are valuation pairs such that the valuation of the product of the seller visited first is smaller but the buyer buys from this seller to avoid the cost of visiting the opponent.} But if both prices are zero, each seller profitably deviates to a small positive price, even though the buyer visits its opponent first. More specifically, it can be shown that when the valuations are uniformly distributed on $[0,1]$, each seller chooses its price randomly from a distribution with support on an interval $[p,P]$ where $0<p<P$. 

To make the model somewhat more tractable, we assume that the buyer faces the Hotelling's travelling cost.  That is, the buyer's location \(x\) is randomly drawn from the uniform distribution on \([0,1]\), and her cost of visiting seller 1 and 2 are $cx$ and \(c(1-x)\), respectively. An alternative modelling choice would be to follow \cite{choi2018consumer} and assume that the buyer's value of each seller's product is the sum of: (i) a component privately known to the buyer, which the buyer learns up front; and (ii) a component that the buyer can learn only when she visits the seller. This alternative is probably more literal but also more involved, and we conjecture that it would lead to the same result.  As \cite{armstrong2016search}, we will keep assuming that the cost of returning to the seller who has been visited is zero.

\subsection{An anti-discriminatory effect of the competition to be visited first}

We argue that the competition for being visited first has an anti-discriminatory effect even when the cost of visiting firms is positive. More specifically, price discrimination is less likely to occur (and if it occurs, its extent is smaller) when firms commit to pricing strategies than when they cannot (before the buyer decides which seller to visit first). The reason is that with commitment sellers have an additional incentive for reducing the price gap between that for first comers and that for returning customers to encourage customers to visit them first. To show this, consider a symmetric equilibrium in which \(\underline{p}<\overline{p}\) of the model without commitment. In such an equilibrium each seller is equally likely to be visited first; that is, seller 1 is visited first if $x<1/2$ and seller 2 is visited first if $x>1/2$.  Without commitment, firms cannot affect the buyer's decision on whom to visit first. But with commitment, by reducing the price on return, each firm can attract  the buyer with \(x\) around \(1/2\) who initially visits its opponent first. If the profit collected from the buyer who visits first is greater than that from the buyer who visits second, such a deviation is profitable. We state this intuition formally in the following proposition:

\begin{proposition}\label{p5} 
Let \(\underline{p}<\overline{p}\) be the prices in a symmetric equilibrium of the model without commitment. Suppose that the buyer with $x=1/2$ prefers to visit a seller to visiting no seller. Suppose further that each seller's profit from $x=1/2$ is strictly greater when the buyer visits the seller first compared to when the buyer visits the seller second. Then, each seller has an incentive to deviate and commit to a slightly lower price upon return. 
\end{proposition}

Before providing the proof, we will discuss the two assumptions in the statement of Proposition \ref{p5}. The assumption that the buyer with $x=1/2$ prefers to visit a seller to visiting no seller makes the analysis interesting. If the buyer with $x=1/2$ prefers to visit no seller, then each seller has its own niche: Seller 1 serves the buyer with small values of $x$, and seller 2 serves the buyer with large values of $x$. Therefore, small changes in the price upon return have no effect. The assumption that each seller's payoff is greater when visited first than when visited second is consistent with the idea of search deterrence. Indeed, a seller can benefit from search deterrence only when the seller is visited first. However, we failed to prove (or disprove) that the assumption is satisfied under general distribution of valuations.  The only known equilibria of the model without commitment are those found numerically (for the uniform distribution of valuations) by \cite{armstrong2013search}. For some of those equilibria, we checked (also numerically) that the assumption is satisfied. For example, for $c=0$, $\underline{p}=0.45$, $\overline{p}=0.51$,\footnote{See Figure 6, panel a on page 34 of  \cite{armstrong2013search}.}  this yields $v^{\ast }=0.71359$;  the profit of the firm visited first is 0.18627 while the profit of the firm visited second is 0.16729. By continuity, this computation guarantees that the assumption is also satisfied for small positive costs $c$. After the proof, we will also show analytically using our result in Section \ref{s4} that our assumption is satisfied for the uniform distribution of valuations and small values of costs $c$.  

\bigskip

\noindent\textbf{Proof:} Assume first that $\overline{p}<1$. When sellers cannot commit to pricing policies, \(\underline{p}<\overline{p}\) is each seller's best response to the same strategy of the opponent. So, the effect of a unilateral small reduction in $\overline{p}$ is of an order lower than 1. Reducing $\overline{p}$ with commitment has the same effects as with no commitment, except an additional effect which was absent in the model with no commitment, namely, it encourages the buyer to visit the deviating seller first, because the reduction of the price on return makes the buyer's loss contingent on return smaller. If a seller slightly reduces the price upon return, then the buyer with $x=1/2$, who was indifferent without commitment regarding which seller to visit first, now prefers to visit the seller who has reduced the price on return. The same is true for the buyer with any $x$ close to $1/2$. Moreover, the buyer returns to the seller visited first with positive probability. So, the buyer's gain from the reduction in $\overline{p}$ is of a first order of the reduction. Thus, the probability of attracting the buyer due to the reduction is also of a first order, and the seller makes a first-order gain from attracting such buyers by the assumption that each seller's payoff is greater when visited first than when visited second.

Suppose now that \(\underline{p}<\overline{p}=1\). By the second assumption, the seller's profit contingent on $x=1/2$ is greater when it is visited first compared to when it is visited second. By unilaterally committing to a slightly lower price upon return, the seller increases the chance that its location will be visited first. The seller does it at no loss, possibly even with a gain, because the seller was not selling to the returning customers before the deviation but may sell to some of them after the deviation. \qed

\bigskip

We will now analytically show that for the uniform distribution of valuations and $c=0$, in the model with no commitment, each seller's payoff is greater when the buyer visits the seller first compared to when the buyer visits the seller second (then by continuity, when \(c\) is sufficiently small, contingent on \(x=1/2\), the profit from being visited first is greater). Indeed, the seller whose location is visited first could deviate to charging all customers \underline{$p$}. Say this is seller 1. This deviation would result in a demand surrounded by the green lines in Figure \ref{fig2}. This demand exceeds the demand when seller 1 is visited second, which is surrounded by the blue lines in Figure \ref{fig2}. In both cases seller 1 charges \underline{$p$} the buyer who buys from the seller. So, the profit after the deviation when the seller is visited first exceeds the profit when the seller is visited second. However, note that before and after this deviation, the profit for seller 1 obtained from being visited second is unchanged. Since no deviation is profitable in equilibrium, the profit from being visited first must be weakly greater before the deviation, which is again higher than the profit from being visited second, and this  gives the result for $c=0$. 

\bigskip 

Thus, the observation that commitment reduces incentives for discriminatory practices holds true even when the cost of visiting is positive, at least under the assumptions of Proposition \ref{p5}. Moreover, if the non-discriminatory equilibrium for $c=0$ is unique, as is the case when the distribution of valuations is uniform, by upper hemi-continuity arguments, if any discrimination at all takes place in equilibrium for small costs, then the discrimination must be small. More precisely, we have the following observation:

\begin{proposition}\label{p6}
Suppose that for $c=0$ there is a unique
symmetric equilibrium in pure strategies in which \(\underline{p}=\overline{p}\). Then, for every $\varepsilon >0$ there is a $\delta >0$ such that if $c<\delta $ then $\left\vert \overline{p}-\underline{p}\right\vert<\varepsilon $ in every symmetric equilibrium in pure strategies.
\end{proposition}

\noindent\textbf{Proof}: It is easy to check that the correspondence that assigns to every $c\geq 0$, the set of symmetric, pure-strategy equilibrium price pairs  \(\overline{p} \geq \underline{p}\)  is upper hemi-continuous. \qed 

\bigskip

So, the only question for the uniform distribution of valuations is whether the additional disincentive for discriminatory policies coming from the competition to be visited first (under commitment) only reduces the extent to which firms discriminate, with the discrimination vanishing for small costs, or it results in no discrimination, at least for small costs.\footnote{One can also argue that with zero cost, the buyer should visit the two firms at the same time, and buy the product with a higher valuation at the price for first comers. Even putting aside the feasibility of this strategy in practice, the contrast between commitment and no commitment in sequential search remains important because of the difference between the outcomes in the settings with small costs.} We partially address this question in the subsequent subsection by showing that there exist non-discriminatory equilibria.

\subsection{Numerical analysis}\label{s53}

For uniform distribution of valuations and a small search cost $c$, our numerical analysis suggests that there is a symmetric equilibrium in which firms offer the same price for first comers and returning customers. More specifically, we analytically show that 
\begin{equation}
p^{\ast }=\sqrt[2]{1+\left( 1+\sqrt[2]{\frac{c}{2}}\right) ^{2}}-\left( 1+%
\sqrt[2]{\frac{c}{2}}\right)  \label{FOC}
\end{equation}
satisfies the first-order condition for a symmetric equilibrium in the auxiliary setting in which firms are only allowed to post single prices. That is, if one firm posts $p^{\ast }$, and the other firm posts $p$, then the derivative of the payoff of the latter firm with respect to $p$ is zero at $p=p^{\ast }$.  Moreover, the only $p^{\ast }$ with this property is given by (\ref{FOC}). We provide the details of this computation in the online appendix.

Our numerical exercise suggests that \underline{$p$}$=\overline{p}=p^{\ast }$ is an equilibrium in the model in which firms are allowed to engage in price discrimination. More specifically, for $c=0.01$, for which (\ref{FOC}) yields $p^{\ast }=0.4$ and a profit of 0.1680000000, we numerically checked the deviations \(\underline{p} \leq \overline{p} \in [0,1]\) belong to the grid of size 0.01, that is, they are numbers with two digits after the decimal point, and \(\underline{p} \leq \overline{p} \in [0.35,0.45]\) belong to the grid of size 0.005. None of them turned out profitable. For the grid of size of 0.01, the best deviation \(\underline{p}=\overline{p}=0.410\) yields the deviating firm a payoff of 0.1679067794 smaller than that from charging the uniform price of 0.4; and for the grid of size 0.005, the best deviation \(\underline{p}=\overline{p}=0.405\) yields the deviating firm a payoff of 0.1679769081 smaller than that from charging the uniform price of 0.4. We enclose the code for this computation in a separate zip file. 

\section{Conclusion}\label{s6}

We study buy-now discount as a search deterrence instrument in a duopoly setting. When the search cost is zero, the possibility of committing to pricing policies, or building reputation thereof makes buy-now discounts ineffective. In equilibrium, firms commit to uniform pricing policies, which stands in contrast with offering buy-now discounts in settings in which firms cannot commit. This difference comes from the fact that commitment leads to competition for being visited as first. Due to the competition, price discrimination is less likely to occur (and if it occurs, its extent is smaller) even when search costs are positve, at least for the uniform distribution of valuations.

Our findings provide an insight on why buy-now discounts, which are offered casually during one-on-one interactions between buyers and sellers, are sometimes observed in certain markets, but not in others, such as real estate. Since such buy-now discounts reduce both consumer surplus and total welfare, our findings imply that commitment is desirable from the welfare perspective; so is the possibility of building reputation.

Finally, we wish to mention that our idea of the competition among firms for being first visited by customers is likely to apply to other settings, such as advertising or obfuscation. We leave this as a promising topic for future research.

\section{Online Appendix}

\subsection{Derivation of $p^{\ast }$ from Section \protect\ref{s53}}

We first characterize the symmetric pure strategy equilibrium when firms are only allowed to set a single price. Let $p^{\ast }$ be the putative equilibrium price. Suppose now firm 1 deviates to $p^{\ast }-\Delta p$. Consider a buyer located at $x$. If she visits firm 1 first, she will buy from firm 1 without visiting firm 2 if her value of firm 1's product exceeds a cutoff $v_{1}(x)$. The buyer will visit firm 2 if $v_{1}<v_{1}(x)$. At $v_{1}=v_{1}(x)$, this buyer is exactly indifferent between buying from firm 1 and continuing search. This cutoff must satisfy the following condition:
\begin{equation*}
v_{1}(x)-p^{\ast }+\Delta p=-c(1-x)+\int_{0}^{1}\max \{v_{2}-p^{\ast},v_{1}(x)-p^{\ast }+\Delta p\}dF(v_{2}),
\end{equation*}
which gives 
\begin{equation*}
v_{1}(x)=1-\Delta p-\sqrt[2]{2c(1-x)}.
\end{equation*}

The buyer will purchase from firm 1 and pay $p^{\ast }-$ $\Delta p$ if $v_{1}>v_{1}(x)$ or $p^{\ast }-$ $\Delta p<v_{1}<v_{1}(x)$ and $v_{1}-p^{\ast}+\Delta p>v_{2}-p^{\ast }$. The expected profit obtained from this event is
\begin{equation*}
\left( p^{\ast }-\Delta p\right) \left[ 1-v_{1}(x)+\frac{1}{2}\left(v_{1}-p^{\ast }+\Delta p\right) \left( p^{\ast }+v_{1}+\Delta p\right)\right],
\end{equation*}
which by using the formula for $v_{1}(x)$ reduces to 
\begin{equation}\label{profit 1}
\left( p^{\ast }-\Delta p\right) \left[ \frac{1}{2}+c(1-x)+\Delta p-\frac{1}{2}(p^{\ast })^{2}\right].  
\end{equation}

If instead the buyer visits firm 2 first, the cutoff $v_{2}(x)$ satisfies the following condition:
\begin{equation*}
v_{2}(x)-p^{\ast }=-cx+\int_{0}^{1}\max \{v_{1}-p^{\ast }+\Delta
p,v_{2}(x)-p^{\ast }\}dF(v_{1}),
\end{equation*}
which gives 
\begin{equation*}
v_{2}(x)=1+\Delta p-\sqrt[2]{2cx}.
\end{equation*}

The buyer will purchase from firm 1 (and pay $p^{\ast }-\Delta p$) if $v_{2}<v_{2}(x)$ and $v_{1}-p^{\ast }+\Delta p>\max \{0,v_{2}-p^{\ast }\}$. Thus the expected profit of firm 1 from this event is 
\begin{equation*}
\left( p^{\ast }-\Delta p\right) \left[ v_{2}(x)\left( 1-p^{\ast }+\Delta
p\right) -\frac{1}{2}\left( v_{2}(x)-p^{\ast }\right) ^{2}\right] \text{,}
\end{equation*}
which by plugging in the expression for $v_{2}(x)$ is equal to 
\begin{equation}\label{profit 2}
\left( p^{\ast }-\Delta p\right) \left[ -cx+\frac{1}{2}+\Delta p+\frac{1}{2}\Delta p^{2}-\frac{1}{2}(p^{\ast })^{2}\right]. 
\end{equation}

Formulas (\ref{profit 1}) and (\ref{profit 2}) allow for computing the total profit of firm 1:
\begin{align*}
&\int_{0}^{x^{\ast }}\left( p^{\ast }-\Delta p\right) \left[ \frac{1}{2}+c(1-x)+\Delta p-\frac{1}{2}(p^{\ast })^{2}\right]\\
+&\int_{x^{\ast }}^{1}\left( p^{\ast }-\Delta p\right) \left[-cx+\frac{1}{2}+\Delta p+\frac{1}{2}\Delta p^{2}-\frac{1}{2}(p^{\ast})^{2}\right]\\
=&\left( p^{\ast }-\Delta p\right) \left[ \frac{1}{2}+\Delta p-\frac{1}{2}(p^{\ast })^{2}+cx^{\ast }+\frac{1}{2}\Delta p^{2}(1-x^{\ast })-\frac{c}{2}\right]
\end{align*}

Notice that the buyer will not randomize over which firm to visit first; instead there exist a cutoff $x^{\ast }$ such that when $x<x^{\ast }$, the buyer (strictly) prefers to visit firm 1 first, and when $x>x^{\ast }$, she (strictly) prefers to visit firm 2 first. At $x^{\ast }$, the buyer is indifferent, i.e., for small $\Delta p$: 
\begin{align*}
&\lbrack 1-v_{1}(x)]c(1-x)+\int_{v_{1}(x)}^{1-\Delta p}\left[\int_{v_{1}+\Delta p}^{1}\left( v_{1}-v_{2}+\Delta p\right) dv_{2}\right]dv_{1}\\
=&[1-v_{2}(x)]cx+\int_{v_{2}(x)}^{1}\left[ \int_{v_{2}-\Delta p}^{1}\left(v_{2}-v_{1}-\Delta p\right) dv_{1}\right] dv_{2},
\end{align*}
which gives 
\begin{equation*}
c(1-x)\Delta p+\frac{2}{3}c(1-x)\sqrt[2]{2c(1-x)}=-cx\Delta p+\frac{2}{3}cx\sqrt[2]{2cx}+\frac{1}{6}(\Delta p)^{3}.
\end{equation*}
The $x^{\ast }$ that solves this equation is a function of $\Delta p$, and by the implicit function theorem, 
\begin{equation*}
\frac{dx^{\ast }}{d(\Delta p)}=\frac{c-\frac{1}{2}(\Delta p)^{2}}{\frac{1}{2}(2c)^{3/2}\left(\sqrt[2]{1-x^{\ast }}+\sqrt[2]{x^{\ast }}\right)};
\end{equation*}
for $\Delta p=0$ and $x^{\ast }=1/2$, this yields $dx^{\ast }/d(\Delta p)=\frac{1}{2\sqrt[2]{c}}$.

Taking the derivative of the total profit of firm 1 with respect to \(\Delta p\) at $\Delta p=0$ gives 
\begin{align*}
&p^{\ast }\left[ 1+c\frac{dx^{\ast }}{d(\Delta p)}\right] -\left[ \frac{1}{2}-\frac{1}{2}(p^{\ast })^{2}+cx^{\ast }-\frac{c}{2}\right]\\
=&p^{\ast }\left( 1+\frac{\sqrt[2]{c}}{2}\right) +\frac{1}{2}(p^{\ast })^{2}-\frac{1}{2},
\end{align*}
and this derivative is equal to $0$ when 
\begin{equation*}
p^{\ast}=\sqrt[2]{1+\left(1+\sqrt[2]{\frac{c}{2}}\right)^{2}}-\left(1+\sqrt[2]{\frac{c}{2}}\right).
\end{equation*}

\subsection{Code for checking the existence of profitable deviations}

All code used to check for the existence of profitable deviations is provided in the zip file submitted with the paper.

\newpage
\bibliographystyle{aer}
\bibliography{references}

\end{document}